%% file: main.tex
\documentclass[sigconf,nonacm]{acmart}

\AtBeginDocument{%
  \providecommand\BibTeX{{%
    \normalfont B\kern-0.5em{\scshape i\kern-0.25em b}\kern-0.8em\TeX}}}

\acmConference[FAccT '21]{FAccT '21: ACM Conference on Fairness, Accountability, and Transparency}{2021}{Virtual}
\acmBooktitle{FAccT '21: ACM Conference on Fairness, Accountability, and Transparency, 2021, Virtual}

\usepackage[utf8]{inputenc}
\usepackage{fullpage}
\usepackage{amsmath, amsthm}
\usepackage{breqn}
\usepackage{hyperref}

\input{comment.tex}

\usepackage{graphicx}
\graphicspath{ {./images/} }
\usepackage{thmtools}
\usepackage{thm-restate}

\usepackage{algorithm}
\usepackage[noend]{algpseudocode}
\usepackage{mathrsfs}
\DeclareMathOperator{\rgt}{rgt}
\DeclareMathOperator{\unf}{unf}
\begin{document}

\title{ProportionNet: Balancing Fairness and Revenue for Auction Design with Deep Learning}

\author{Kevin Kuo}
\authornote{Equal contribution}
\email{kkuo1@umd.edu}
\affiliation{University of Maryland}

\author{Anthony Ostuni}
\authornotemark[1]
\email{aostuni@umd.edu}
\affiliation{University of Maryland}

\author{Elizabeth Horishny}
\authornotemark[1]
\email{ehorishny1@pride.hofstra.edu}
\affiliation{Hofstra University}

\author{Michael J. Curry}
\authornotemark[1]
\email{curry@cs.umd.edu}
\affiliation{University of Maryland}

\author{Samuel Dooley}
\authornotemark[1]
\email{sdooley1@cs.umd.edu}
\affiliation{University of Maryland}

\author{Ping-yeh Chiang}
\authornotemark[1]
\email{pchiang@cs.umd.edu}
\affiliation{University of Maryland}

\author{Tom Goldstein}
\email{tomg@cs.umd.edu}
\affiliation{University of Maryland}

\author{John Dickerson}
\email{john@cs.umd.edu}
\affiliation{University of Maryland}

\begin{abstract}
The design of revenue-maximizing auctions with strong incentive guarantees is a core concern of economic theory.  Computational auctions enable online advertising, sourcing, spectrum allocation, and myriad financial markets. Analytic progress in this space is notoriously difficult; since Myerson's 1981 work characterizing single-item ``optimal'' auctions, there has been limited progress outside of restricted settings. A recent paper by D{\"u}tting et al.\ circumvents analytic difficulties by applying deep learning techniques to, instead, approximate optimal auctions. In parallel, new research from Ilvento et al.\ and other groups has developed notions of fairness in the context of auction design. Inspired by these advances, in this paper, we extend techniques for approximating auctions using deep learning to address concerns of fairness while maintaining high revenue and strong incentive guarantees.
\end{abstract}

\maketitle
\pagestyle{plain}

\section{Introduction}
Auctions connect buyers and sellers to enable the exchange of money for goods and services.  Auction theory has a rich history in economics and, more recently, computer science.
Since 1994, the US Federal Communications Commission (FCC) has periodically run multi-billion dollar auctions to allocate electromagnetic spectrum broadcasting licenses requiring immense computational resources \citep{Leyton-Brown17:Economics}.  Technology giants such as Google, Facebook, and Baidu rely heavily on sophisticated auction-based advertising ecosystems to drive the majority of their revenue~\citep{Edelman07:Internet}. Additionally, websites such as eBay and Alibaba's Taobao operate as platforms that connect buyers and sellers, often through auctions.  In aggregate, the contribution to the world economy of computational auctions is measured in the hundreds of billions, if not trillions, of dollars per year~\citep{Alphabet20:Q2,Facebook20:Q1,Baidu20:Q1}. 

The design of auctions is thus quite important. In all cases described, the rules for determining winners and payments from bids are carefully designed to make sure the auctions fulfill desirable properties. This is a major focus of the broader field of mechanism design~\citep{Roughgarden10:Algorithmic}.

In the typical theoretical model for auction mechanisms, players are presumed to have some private valuations of the items up for sale, which are drawn from some publicly-known distribution. The players then place their bids, possibly choosing to strategically lie while trying to anticipate the strategic behavior of others. Typically players are assumed to be rational, so that in this setting they will choose actions from a Bayes-Nash equilibrium, but in reality this equilibrium may be very complex and difficult for the designer and players to determine.

One solution to this problem is to focus on strategyproof, or incentive compatible, auctions. These are auctions where, even though players are free to lie about their private valuations, rational players will simply choose to tell the truth. Subject to this constraint, equilibrium play is simple, and the mechanism designer can focus on ensuring other desirable properties. The classic strategyproof auction is known as the Vickrey-Clarke-Groves (VCG) auction \cite{Vickrey61:Counterspeculation,Clarke71:Multipart,Groves73:Incentives}, which has the additional desirable property of maximizing social welfare (i.e. the total utility enjoyed by all auction participants).

 Auction designers often care about social welfare, but in many cases an auctioneer selling items may instead wish (or have an obligation, as in auctions of spectrum and other goods belonging to the public) to maximize their own revenue, subject to strategyproofness. Myerson's \cite{Myerson1981} groundbreaking work defined the optimal strategyproof auction for selling a single item, but progress has been limited in characterizing strategyproof, revenue-maximizing auctions beyond this setting. While there are some results for selling multiple items to a single bidder \cite{daskalakis2017strong,manelli2006bundling,pavlov2011optimal}, even for selling just two items to two bidders, no results are known.
 
\begin{figure*}[t]
    \centering
    \includegraphics[scale=0.50]{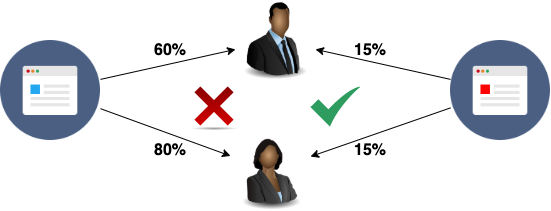}
    \caption{While RegretNet appears successful at determining revenue-maximizing allocations, it is blind to the fairness of such allocations. Thus, there may be a high probability of unequal allocations, even for two items that are equivalent for all meaningful purposes. Concretely, it may allocate an online advertisement for a career opportunity to an equally qualified man and woman at notably different proportions (as on the left). Our network ProportionNet prevents such unfairness; (as on the right) it forces similar advertisement allocation proportions between the two similar individuals.}
    \label{fig:mvpayment}
\end{figure*}

The persisting challenge of designing optimal auctions, and the fact that typical theoretical assumptions involve a probability distribution over valuations, have resulted in attempts to formulate the auction design problem as a machine learning problem. In particular, \citet{dutting2019optimal} use neural networks to represent an auction mechanism (as a function from a vector of bids to outcomes), and define a learning objective to enforce strategyproofness while encouraging revenue maximization. %

In addition to their economic importance, the design of auctions can also have serious social impact. With this in mind, what properties, in addition to strategyproofness, might the designer of a revenue-maximizing auction have reason to enforce? \textbf{A major concern must be fairness with respect to protected characteristics.}

Consider the case of online advertising---one of the most important real-world applications for the theory of mechanism design. When placing advertisements in certain categories (job ads, ads for certain financial services, ads for housing, etc.), companies have a legal obligation to avoid discrimination on the basis of protected characteristics such as race, gender, and national origin. Yet a 2015 study \cite{AutomatedExperimentsonAdPrivacySettings} showed a difference between employment advertisements received by male and female users: male users were shown advertisements promising higher salaries than female users. Furthermore, \cite{Ali_2019} observe that Facebook's preemptive categorization of ad-user relevancy skews ads toward certain genders and racial groups.

A number of papers have considered the mechanism design problem when fairness with respect to protected characteristics is required \cite{ilventoTVfair,chawla2019multi,chawla2020fairness}. Here, whatever the exact details of how fairness is defined, the notion of fairness does not consider the bidders in the auction, but rather the individuals whose impressions are the ``items'' for sale -- thus imposing fairness constraints means imposing constraints on the allocations made by the auction mechanism.

Defining a strategyproof, fair auction that maximizes revenue remains theoretically challenging. 
Recent work in this area includes \citep{Celis19:Toward}, which provides theory and algorithms for finding optimal  itemwise Myerson auctions under fairness constraints. \citet{Nasr20:Bidding} computes fair strategies from the bidder's point of view.

As with other auctions, we see the use of machine learning as a way out of this impasse -- we aim to extend the techniques of \cite{dutting2019optimal} and others to allow the imposition of fairness constraints on learned mechanisms. Doing so allows insights into the cost to revenue of imposing fairness, and the structure of fair mechanisms in some settings.

\noindent\textbf{Our contributions.}  
\begin{itemize}
    \item We provide a deep-learning-based method for designing approximately fair, strategyproof, revenue-maximizing auctions given access to samples from the valuation distribution. Our approach extends the RegretNet approach \cite{dutting2019optimal} with fairness constraints (and preserves its generalization guarantees): a melding of ideas from the fair ML and economics \& computation communities.
    
    \item This represents a step towards the larger problem of designing revenue-maximizing multi-item auctions under not only strategyproofness but also fairness constraints, potentially motivating and informing future theoretical work.
\end{itemize}

\section{Background}

We first describe the typical formal model of auction design, the challenges of designing good auctions, and the use of deep learning techniques to circumvent these challenges. Then we discuss problems of fairness in auctions and describe one formal definition of fairness that makes sense in an auction setting.

\subsection{Auction Model}
An auction involves a set of agents $N = \{1,\dots,n\}$ bidding for items $M = \{1,\dots,m\}$. Each agent $i\in N$ has a corresponding valuation function $v_i$. These private valuations are presumed to be drawn at random from a publicly known distribution $V_i$ of their possible valuation functions. We denote a profile of the $n$ valuation functions as $v = (v_1,...v_n)$.

Let $v_i(S)$ represent the agent's value for a subset of items $S \subseteq M$. In the most general case $v_i$ is defined for all subsets of $M$; these are known as \textbf{combinatorial} valuations. In practice combinatorial valuations are very difficult to deal with, as each user must report $2^M$ different bids, and this quantity may grow unreasonably large.

One can instead use simpler familes of valuations. With \textbf{additive} valuations, an agent's valuation for a subset of items is $v_i(S) = \sum_{j\in S}v_i(\{j\})$, the sum of the individual items' valuations. With \textbf{unit-demand} valuations, the value of a subset is $v_i(S) = \max_{j\in S}v_i(\{j\})$, the maximum individual valuation within that subset. Both cases reduce the input space of $v_i$ from size $2^M$ to $M$ -- users need only bid on each item. In this work, we will operate with either the additive or unit-demand assumption, but \textit{not} with combinatorial valuations.

Given their private valuations, each agent reports a bid vector $b_i$ to the auctioneer. Note that $b_{ij}$, agent $i$'s bid on the $j^{\text{th}}$ item, is not necessarily $v_i(\{j\})$; our auction operates under the assumption that agents are free to report bids which do not represent their true item valuations (although we will try to discourage them from doing so).

Finally, based on the profile of bids $b = (b_1,\dots,b_n)$, the auction determines an outcome using the allocation and payment rules $g(b): \mathbb{R}^{mn} \rightarrow [0,1]^{nm}$ and $p(b): \mathbb{R}^{mn} \rightarrow \mathbb{R}^{n}$. We will refer to the matrix of allocation probabilities as $g(b) = z$, and the allocation probability $g(b)_{i,j}$ of item $j$ to agent $i$ as $z_{i,j}$. Allocation probabilities for each item must sum to 1 (items cannot be overallocated). Additionally for unit-demand auctions, we restrict the allocation to allow each bidder to win, in expectation, at most 1 item.

Given the allocation, each agent receives a utility equal to their true valuation of the items they win, minus their payment. For additive and unit-demand bidders, this can be represented in linear form as $u_i = \sum_j v_{i,j}z_{i,j} - p_i$.

\subsection{Desirable Auction Properties}

 A mechanism is \textbf{individually rational} (IR) when an agent is guaranteed non-negative utility: $u_i(v_i; v) \geq 0$ $\forall i \in N, v\in V$ -- the agent will never be made to overpay for what they win (assuming they bid truthfully). A mechanism is \textbf{dominant-strategy} \textbf{incentive-compatible} (DSIC) or \textbf{strategyproof} if every agent maximizes their own utility by bidding truthfully, regardless of the other agents' bids. To define this notion formally, it is useful to first define the notion of \textbf{regret}, which is the difference in utility between the bid player $i$ actually made (for our purposes, typically a truthful bid) and the best possible strategic bid:
\begin{equation}
    \rgt_i(v) = \max_{b_i} u_i(b_i, v_{-i}) - u_i( v_i, v_{-i})
    \label{eq:regret}
\end{equation}

An auction is DSIC when regret (for a truthful bid) is always zero for every player -- they have no incentive to do anything other than tell the truth.
 
In addition to satisfying the IR and DSIC constraints, the auctioneer seeks to \textbf{maximize their expected revenue}. If the auction is truly DSIC, we assume players will bid truthfully, and as a result revenue is simply $E_{v\sim V}[\sum_{i \in N}p_i(v)]$.

\subsection{Optimal Auction Design}

Myerson's seminal 1981 work on auctions settled the question of revenue-maximizing, strategyproof auction design in the common setting of selling a single item \cite{Myerson1981}. However, little analytical progress has been made since then, outside of partial results in restricted settings (many involve selling multiple items to only a single bidder) \cite{manelli2006bundling, pavlov2011optimal, giannakopoulos2014duality, daskalakis2017strong, yao2017dominant}. %
Because deriving analytic solutions to mechanism design problems has been so difficult, another trend within the research community has been to approximate mechanisms by formulating the mechanism design problem as a learning problem -- the learned solutions may be adequate in their own right, as well as providing a starting point for theoretical investigation. 

\subsection{Optimal Auction Design Through Deep Learning}

D{\"u}tting et al. published work on RegretNet, a neural network architecture that models an auction mechanism \cite{dutting2019optimal} -- this work has been extended and applied in other areas \cite{shen2019automated, feng2018deep}. The core idea of RegretNet is that in the Bayesian auction setting, one knows the valuation distributions from which samples can presumably be drawn, and the allocation and payment rules are just functions, so one can use neural networks as function approximators, with a learning objective designed to maximize revenue while enforcing strategyproofness.

In detail, the allocation and payment functions $(g(b), p(b))$ are represented as neural networks $(g^w(b), p^w(b))$ where $w$ is the set of learned weights. These networks are standard feedforward networks. The allocation networks $g$ end with a softmax layer, to ensure that allocations are valid categorical distributions (additionally in the unit-demand setting, that each player is allocated a single item). The payment network ends with a sigmoid layer, outputing a value $\tilde{p}_i$ in $[0,1]$ for each bidder; given the allocations z, the final payment $p_i = \tilde{p}_i\left(\sum_j z_{i,j} v_{i,j}\right)$. This ensures individual rationality cannot be violated.

The training data for RegretNet is a dataset of $L$ bid profiles sampled from the valuation distribution $V$; these are used for training by standard gradient descent. The goal is to maximize the payments drawn from truthful bids, subject to strategyproofness. Maximizing expected payment can be done by simply maximizing the mean payment over training samples; enforcing strategyproofness is more difficult. The authors of RegretNet relax the notion of strict dominant-strategy incentive compatibility to a slightly weaker notion of expected regret: $\mathbb{E}_v\left[\sum_i \rgt_i(v)\right]$ -- note if this is exactly zero, then the mechanism is truly DSIC.

To estimate the regret under a specific valuation, the authors of \cite{dutting2019optimal} perform gradient ascent on the network inputs to find a nontruthful bid that maximizes player utility -- this is a quantity they call $\widehat{\rgt}_i$.

To enforce the regret constraint, RegretNet uses the augmented Lagrangian method and incorporates a set of Lagrange multipliers $\lambda = \{\lambda_1, ..., \lambda_n\}$ and a quadratic parameter $\rho$.
\begin{multline}
\mathcal{C}_\rho(w;\lambda) = 
-\frac{1}{L} \sum_{\ell =1}^{L} \sum_{i \in N} p_i^{w}(v^{(\ell)})\\ +
\sum_{i \in N} \lambda_i \widehat{rgt_i}(w) +
\frac{\rho}{2}(\sum_{i \in N}\widehat{rgt_i}(w))^2.
\label{eq:regretnetloss}
\end{multline}

The training procedure involves alternating gradient steps to solve $\min_{w} \max_{\lambda} C_\rho(w; \lambda)$, as well as 25 gradient ascent steps at each iteration to approximate $\widehat{\rgt}_i$. At test time, revenue is evaluated on new samples and $\widehat{\rgt}_i$ is approximated using 1000 gradient ascent steps.

\subsection{Fairness} 

First, we discuss real-world examples of unfairness in advertising auctions. Here, the unfairness is suffered by the individuals whose ad impressions are the ``items'' up for auction. We then describe one attempt from the literature to formalize fairness, which we will adopt as an additional constraint in the RegretNet approach to auction learning.

\subsubsection{Unfairness in Real-World Ad Auctions}

The ad auctions currently in place throughout the internet have been shown to produce discriminatory ad allocations. A core feature of online ad platforms is the ability to target users with certain properties. Thus, in an online ad allocation, platforms will typically consider additional factors in tandem with advertiser bids. For instance, key components include the demographics of their users and the target audience of the advertiser.  Google's and Facebook's platforms take in user attributes (such as location, device type, and search query) as well as advertisement relevance in each auction \cite{googleADs,facebookBusi}.

 The practice of targeting advertisements has in the past lead to discriminatory allocations. For example, an experiment in 2013 \cite{sweeney_2013} publicized the disproportionate likelihood of receiving online ads related to arrest records with a search query of a black-sounding name in contrast with a white-sounding name, even when the advertiser's preferred search queries and bids represented white and black sounding names equally. Additionally, \cite{AutomatedExperimentsonAdPrivacySettings} showed that between female and male users, male users with the same Google search queries tended to receive advertisements for higher-paying job offers than female users.

Facebook's auctions are similar to Google's, with an extended focus on user targeting, with over 2,000 differentiating user categories, including location, age, and income. A study in 2018 showcased the immensity of Facebook's resources, with the proven ability to target users by the single-person and single-household level. As highlighted in the study, this not only violates user privacy, but could put users in vulnerable locations, such as cancer treatment facilities, Planned Parenthood, and rehab centers, at risk \cite{faizullabhoy_a_2018}. Furthermore, \cite{Ali_2019} describes an automated advertisement classification system which determines the ideal demographic for an ad regardless of an advertiser's preferences. This feature has lead to discriminatory allocations based on users' race and gender for ads such as jobs and housing.

Finally, research exploring the disparity of advertisements of STEM job opportunities between male and female users has concluded that Facebook's determination of user prices could lead to discriminatory allocations. In other words, a low-bidding job advertisement that intends to advertise to all users regardless of gender may win more allocations with male users, because Facebook rates female users as more expensive, as women have been noted to interact with advertisements more \cite{lambrecht_tucker_2016}. Ongoing lawsuits regarding Facebook's discriminatory advertisement mechanism confirm unfairness within online ad auctions is a real concern \cite{merrill_2020}.

In all these cases, as mentioned above, fairness is with respect to the ads served to the users, corresponding to ``items'' in the typical model of auctions. (We emphasize this to distinguish our case from the more typical problem of fair mechanism design, where one is concerned with a fair allocation for the mechanism participants.) To mathematically formalize a notion of fairness in this context, we utilize the definition of \textbf{total variation fairness} from \cite{ilventoTVfair} due to its generality.
\subsubsection{Formalizing Unfairness}

Let $C = \{C_1,...,C_c\}$, denote a partition of the set $[n]$ of advertisers, or agents, into c categories. For $1 \leq k \leq c$, let $d^k: M \times M \rightarrow [0, 1]$ define a distance metric between all pairs of users, or items.  The auction mechanism satisfies \textbf{total variation fairness} if the $\ell_1$-distance between allocations (summed over a subset of advertisers $C_k$) for any two users is at most the distance between those users. That is, total variation fairness is satisfied when 
\begin{equation}
  \forall k \in \{1, ..., c\}, \forall j, j' \in M, \sum_{i\in C_k}|z_{i,j} - z_{i,j'}| \leq d^k(j,j').
    \label{eq:tvfair}
\end{equation}
In other words, similar users cannot be treated too differently, although the degree of permissible different treatment might be tighter (for instance, for job or housing advertisements) or looser. For example, if $d^k(j, j')$ were simply defined to be constant, this would disallow allocations in which one item is allocated significantly more than another. If certain advertisers need not worry about unfairness, they might be put into a different category $C_k$ with looser constraints. Likewise if unfairness is less of a concern between certain pairs of items, their distance could be greater, allowing more disparity in allocations.

\subsection{Our Work} 

Our research is an amalgamation of deep learning techniques and fairness concerns; we propose a machine learning solution to find an auction that is DSIC, IR, revenue maximizing, and fair. Our work extends the RegretNet architecture to satisfy the total variation fairness constraint between all pairs of auction items.

\section{Methodology}

\subsection{Fairness Constraint}

To adapt the definition \ref{eq:tvfair} for use as a neural loss function, we define \textbf{unfairness} as a measure of how much the total variation constraint is violated by an auction allocation. The unfairness experienced by a user $j$ is:
\begin{equation}
\small \hspace{-2mm} %
    \unf_j = \sum_{{j'}\in M}\sum_{C_k\in C}\max(0, (\sum_{i\in C_k}\max(0, z_{i,j} - z_{i,j'})) - d^k(j,j'))
\label{eq:unfloss}
\end{equation}
A sum of unfairness over all users $\overline{\unf} = \sum_{j\in M} \unf_j$ allows us to quantify how unfair an auction outcome is for all users involved. 
\subsection{Network Architecture}
We use the same additive and unit-valuation network architectures as RegretNet for arbitrary numbers of agents and items. We enforce our fairness constraint using the augmented Lagrangian approach in RegretNet by incorporating an additional set of multipliers $\lambda_f$. Our modified loss function $\mathcal{C}_{\rho}(w;\lambda)$ is defined as:
\begin{equation}
\begin{aligned}
\mathcal{L}_{\rgt} &= \sum_{i \in N} {\lambda_{(r,i)}} \rgt_i(w) + \frac{\rho_r}{2}(\sum_{i \in N}\rgt_i(w))^2 \\
\mathcal{L}_{\unf} &= \sum_{j \in M} \lambda_{(f,j)}\unf_j(w) + \frac{\rho_f}{2}(\sum_{i \in M}\unf_j(w))^2 \\
\mathcal{C}_\rho(w;\lambda) &= -\frac{1}{L} \sum_{l=1}^{L} \sum_{i \in N} p_i^{w}(v^{(l)}) + \mathcal{L}_{\rgt} + \mathcal{L}_{\unf}
\end{aligned}
\label{eq:propotionnetloss}
\end{equation}

\subsection{Training Procedure}

The procedure closely matches that of RegretNet, with the addition of updating $\lambda_f$ and an additional quadratic parameter $\rho_f$ for our fairness penalty.
\begin{algorithm}
\caption{ProportionNet Training}\label{euclid}
\begin{algorithmic}[1]

\State \textbf{Input:} Minibatches $S_1, ..., S_T$ of size $B$

\State \textbf{Parameters:} $\rho^t_r, \rho^t_f, \gamma, \eta \in \mathbb{R}_{\geq 0}$ $\Gamma \in \mathbb{N}$

\State \textbf{Initialize:} $w^0 \in \mathbb{R}^d, \lambda^0_r \in \mathbb{R}^n, \lambda^0_f \in \mathbb{R}^m$

\For{$t = 0$ to $T$}
    \State Receive minibatch $S_t = \{v^{(1)}, ..., v^{(B)}\}$
    \State Initialize misreports $v_i^{\prime(\ell)}\in V_i,\forall \ell \in [B],i\in N$
    \For{$\gamma = 0$ to $\Gamma$}
        \For{$\ell \in [B], i \in N$}
            \State ${v'}^{(\ell)}_i \leftarrow {v'}^{(\ell)}_i +  \gamma\nabla_{v{'}_i}u_i^w{(v_i^{(\ell)}}; (v{'}_{i}^{(\ell)}, v_{-i}^{(\ell)}))$
        \EndFor
    \EndFor
    \State Compute Lagrangian gradient and update $w^t$:
    \State $w^{t+1} \leftarrow w^t - \eta\nabla_w\mathcal{C}_{\rho_t}(w^t, \lambda^t_r, \lambda^t_f)$
    \State Update Lagrange multipliers every $Q_r$ (regret) and $Q_f$ (fairness) iterations:
    \If{$t$ is a multiple of $Q_r$}
        \State $\lambda_{(r,i)}^{t+1} \leftarrow \lambda_{(r,i)}^t + \rho_r^t \widetilde{rgt}_i(w^{t+1}), \forall i \in N$
    \Else 
        \State $\lambda_{(r,i)}^{t+1} \leftarrow \lambda_{(r,i)}^t$
    \EndIf

    \If{$t$ is a multiple of $Q_f$}
        \State $\lambda_{(f,i)}^{t+1} \leftarrow \lambda_{(f,i)}^t + \rho_f^t \widetilde{unf}_i(w^{t+1}), \forall i \in M$
    \Else 
        \State $\lambda_{(f,i)}^{t+1} \leftarrow \lambda_{(f,i)}^t$
    \EndIf
    
\EndFor
\end{algorithmic}
\end{algorithm}

\subsection{Generalization Bound}

When measuring expected unfairness, we cannot directly compute the expected value---instead, we must estimate it from samples of individual valuation profiles. Similarly to \citet{dutting2019optimal}, we wish to bound the generalization error when estimating auction unfairness from samples---hewing closely to techniques and definitions presented there, we do this in terms of the covering number of the class of auctions, showing that with high probability, our sample estimate is a good upper bound of true expected unfairness.

\begin{restatable}{theorem}{bound} Let $\mathcal{M}$ be a class of auctions that satisfy individual rationality and have $\ell_{\infty, 1}$ covering number $\mathcal{N}_{\infty}(\mathcal{M}, \cdot)$. Fix $\delta \in (0, 1)$. With probability at least $1 - \delta$ over a draw of $L$ valuation profiles, for any $(g^w, p^w) \in \mathcal{M}$, 

\begin{dmath*}
    \mathbb{E}_{v}\left[\sum_{j=1}^m \unf_j \circ g^w(v)\right] \leq \frac{1}{L} \sum_{\ell=1}^{L} \sum_{j=1}^{m} \unf_j \circ g^w(v^\ell) + 2\Delta_L + 4C\sqrt{\frac{2\log(4/\delta)}{L}}
\end{dmath*}

where $C$ is a constant and 

\begin{equation*}
    \Delta_L = \inf_{\epsilon > 0}\left((nm^2 + \epsilon)\sqrt{\frac{2 \log \mathcal{N}_\infty(\mathcal{M}, \frac{\epsilon}{2m^3})}{L}} + \epsilon\right) .
\end{equation*}
\end{restatable}

\section{Experiments}

To experimentally test ProportionNet, we train it in different auction settings with known valuation distributions. Following \cite{dutting2019optimal}, we consider settings involving selling to one agent where revenue-maximizing solutions are known, and additionally add fairness constraints. We then consider settings with more agents and items beyond the reach of theory. Finally, we consider additional settings where the tradeoffs between fairness and bidder preferences are more complex.

\subsection{Experimental Parameters}
For each configuration of $n$ agents and $m$ items, we trained ProportionNet for a maximum of 120 epochs using 640,000 training samples. We used two hidden layers for settings A and B and three for D, E, and F. The hidden layers for setting C are shown in Table \ref{tab:crevenue}. All networks used 100 hidden nodes per layer. We incremented both $\rho_r$ and $\rho_f$ every two epochs and $\lambda_r$ and $\lambda_f$ every 100 iterations. Finally, we used the Adam optimizer for training.

\subsection{The Manelli-Vincent and Pavlov Auctions}
\citet{dutting2019optimal} successfully reproduced the analytic solutions where they were known using the RegretNet framework. These settings are as follows:

A. Single-bidder with additive valuations over two items. Item values are independent draws from $U[0, 1]$ \cite{manelli2006bundling}.

B. Single-bidder, unit-demand valuations over two items. Item values are independent draws from $U[2,3]$ \cite{pavlov2011optimal}.

We train on these settings and additionally apply a uniform fairness constraint to both of these settings. The entries of the total variation fairness distance matrix $D$ are all set to a constant $d(j, j') = d$ for all pairs of users $j$ and $j'$. 

 Figures \ref{fig:mvalloc} and \ref{fig:pvalloc} show the allocation probabilities under a given bid by the single bidder. Training using $d=1$ approximates the revenue-maximizing auction (as this is just standard RegretNet). Training with $d=0$ results in an auction where both items are always allocated with equal probability. Figure 3 shows the training curves for revenue as ProportionNet is trained on different magnitudes of fairness -- note that the decrease in revenue over time reflects the network learning to enforce the constraints, but that stronger fairness constraints result in lower expected revenue.

\begin{figure}
    \centering
    \includegraphics[scale=0.20]{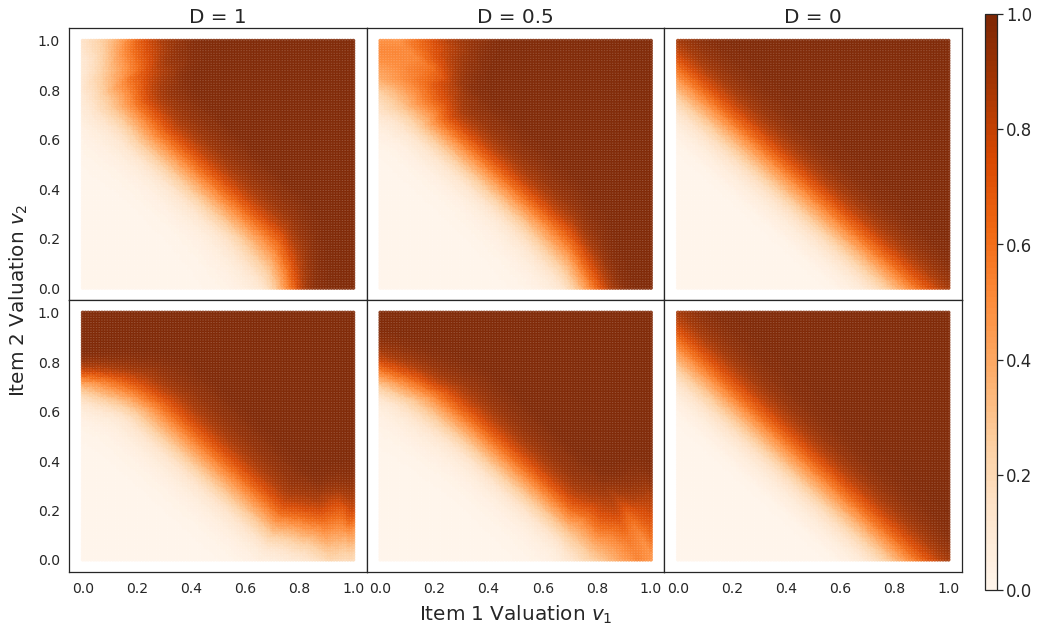}
    \caption{Setting A allocation rule after training with varying values of $D$. Rows 1 and 2 indicate allocation probabilities for Item 1 and 2, respectively.}
    \label{fig:mvalloc}
\end{figure}

\begin{figure}
    \centering
    \includegraphics[scale=0.20]{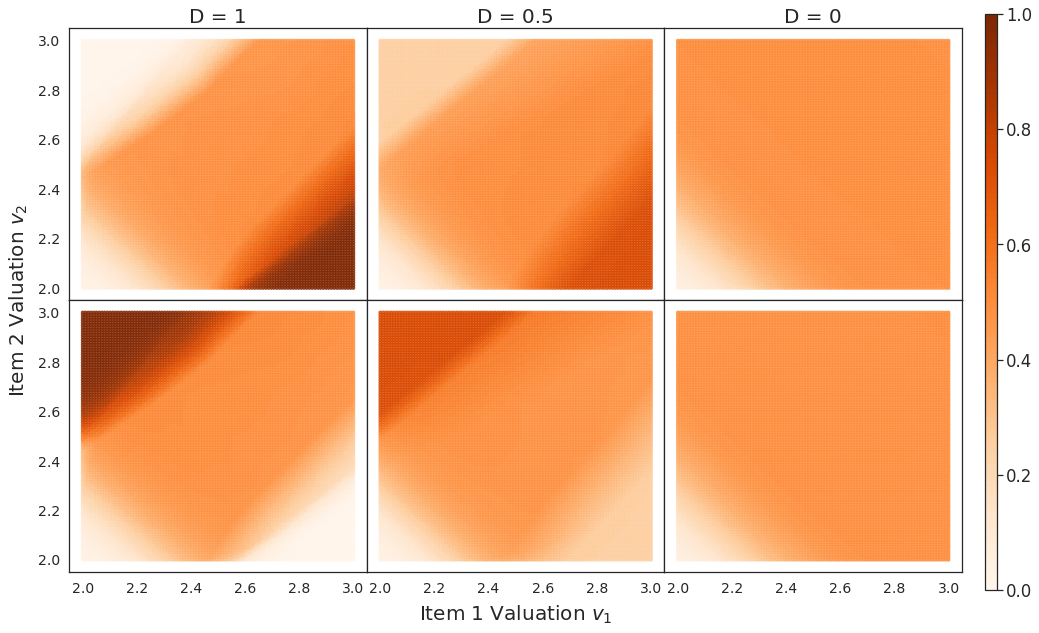}
    \caption{Setting B allocation rule after training with varying values of $D$. Rows 1 and 2 indicate allocation probabilities for Item 1 and 2, respectively.}
    \label{fig:pvalloc}
\end{figure}

\begin{figure}
    \centering
    \includegraphics[scale=0.25]{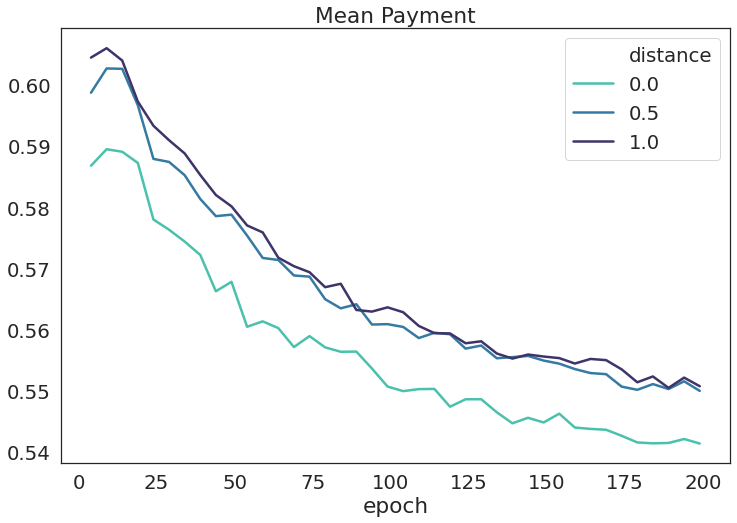}
    \caption{Expected revenue as ProportionNet is trained with different degrees of fairness on the Manelli-Vincent (setting A) auction.}
    \label{fig:mvpayment}
\end{figure}

\begin{table*}[t]
    \centering
\begin{tabular}{|c|c|c|c|c|c|c|c|}
\hline
\multicolumn{8}{|c|}{Sweep Revenue, Mean (StDev)} \\ \hline
n x m & $\ell$ & Myr & \multicolumn{5}{|c|}{D}          \\ \hline
-     & -    & - & 1.00 & 0.75 & 0.50 & 0.25 & 0.00 \\
\hline
1 x 2 & 2 & 0.50 & 0.546 (0.369) & 0.545 (0.363) & 0.544 (0.367) & 0.541 (0.353) & 0.538 (0.374) \\
1 x 3 & 2 & 0.75 & 0.858 (0.504) & 0.856 (0.491) & 0.844 (0.482) & 0.845 (0.474) & 0.845 (0.484) \\
1 x 4 & 2 & 1.00 & 1.199 (0.602) & 1.191 (0.579) & 1.182 (0.591) & 1.189 (0.591) & 1.180 (0.589) \\
1 x 5 & 2 & 1.25 & 1.540 (0.720) & 1.532 (0.714) & 1.533 (0.711) & 1.529 (0.698) & 1.534 (0.720) \\
1 x 6 & 2 & 1.50 & 1.877 (0.832) & 1.895 (0.801) & 1.892 (0.805) & 1.892 (0.811) & 1.886 (0.819) \\
\hline
2 x 2 & 2 & 0.83 & 0.865 (0.349) & 0.858 (0.343) & 0.852 (0.334) & 0.838 (0.327) & 0.830 (0.323) \\
2 x 3 & 2 & 1.25 & 1.234 (0.399) & 1.221 (0.388) & 1.215 (0.382) & 1.209 (0.383) & 1.206 (0.382) \\
2 x 4 & 2 & 1.67 & 1.720 (0.462) & 1.709 (0.452) & 1.683 (0.444) & 1.680 (0.438) & 1.681 (0.438) \\
2 x 5 & 3 & 2.08 & 2.194 (0.525) & 2.168 (0.509) & 2.138 (0.496) & 2.135 (0.492) & 2.133 (0.489) \\
2 x 6 & 3 & 2.50 & 2.665 (0.570) & 2.631 (0.552) & 2.594 (0.537) & 2.594 (0.533) & 2.591 (0.531) \\
\hline
3 x 2 & 2 & 1.06 & 1.056 (0.287) & 1.036 (0.280) & 1.022 (0.274) & 1.010 (0.269) & 1.004 (0.268) \\
3 x 3 & 2 & 1.59 & 1.546 (0.344) & 1.533 (0.336) & 1.502 (0.326) & 1.499 (0.323) & 1.500 (0.324) \\
3 x 4 & 3 & 2.12 & 2.081 (0.404) & 2.040 (0.386) & 2.011 (0.374) & 2.003 (0.373) & 2.001 (0.371) \\
3 x 5 & 3 & 2.66 & 2.577 (0.443) & 2.540 (0.426) & 2.510 (0.416) & 2.498 (0.413) & 2.495 (0.412) \\
3 x 6 & 3 & 3.19 & 3.073 (0.481) & 3.048 (0.471) & 3.011 (0.462) & 3.013 (0.458) & 3.007 (0.456) \\
\hline
4 x 2 & 2 & 1.23 & 1.209 (0.302) & 1.177 (0.276) & 1.151 (0.256) & 1.135 (0.242) & 1.129 (0.244) \\
4 x 3 & 3 & 1.84 & 1.769 (0.286) & 1.720 (0.284) & 1.678 (0.286) & 1.645 (0.293) & 1.642 (0.294) \\
4 x 4 & 3 & 2.45 & 2.247 (0.302) & 2.204 (0.307) & 2.156 (0.334) & 2.159 (0.334) & 2.158 (0.336) \\
4 x 5 & 4 & 3.06 & 2.819 (0.426) & 2.765 (0.394) & 2.707 (0.375) & 2.692 (0.377) & 2.693 (0.379) \\
4 x 6 & 4 & 3.68 & 4.284 (0.491) & 3.279 (0.421) & 3.246 (0.413) & 3.306 (0.399) & 3.007 (0.447) \\
\hline
5 x 2 & 3 & 1.34 & 1.307 (0.295) & 1.305 (0.242) & 1.268 (0.228) & 1.239 (0.221) & 1.230 (0.224) \\
5 x 3 & 3 & 2.02 & 1.906 (0.360) & 1.894 (0.301) & 1.824 (0.273) & 1.777 (0.263) & 1.767 (0.269) \\
5 x 4 & 4 & 2.69 & 2.389 (0.326) & 2.381 (0.349) & 2.271 (0.305) & 2.291 (0.314) & 2.318 (0.314) \\
5 x 5 & 4 & 3.36 & 3.670 (0.402) & 2.906 (0.358) & 2.876 (0.353) & 2.851 (0.348) & 2.836 (0.352) \\
5 x 6 & 5 & 4.03 & 3.489 (0.618) & 3.467 (0.417) & 3.491 (0.420) & 3.445 (0.493) & 3.385 (0.370) \\
\hline
\end{tabular}
\caption{Setting C Revenue -- auctions with $U[0,1]$ valuations for $n$ bidders and $m$ items.}
\label{tab:crevenue}
\end{table*}

\subsection{Scaling Up}
Next, we experiment with larger auctions where there may be no viable analytical solution, even without considering fairness constraints. We define setting C:

C. $n$ bidders with additive valuations over $m$ items. All values (regardless of bidder or item) are independent draws from $U[0, 1]$.

We tested all combinations of $n=1,...,5$ bidders, $m=2,...,6$ items, and fairness constraints $d=0,0.25,..., 1.00$.

Table \ref{tab:crevenue} shows the results. $\text{Myr}$ denotes expected revenue of the itemwise Myerson auction -- selling each item independently in a strategyproof auction -- while $\ell$ denotes the number of hidden layers used for training (100 hidden nodes per layer). More detailed charts with individual regret and unfairness values and standard deviations are in the appendix.

Note that as the number of agents and items increases, it becomes increasingly difficult to both maximize revenue and obey the regret and unfairness constraints. Thus, some of our results (which were primarily selected on the criterion of low mean and standard deviation of regret and unfairness) do not exceed the baseline itemwise Myerson revenue. Additionally, results for the 4 x 6 and 5 x 5 auctions significantly exceed the itemwise Myerson, but their regret and unfairness values are also high. However, prior success in applying the augmented Lagrangian to higher-complexity auctions (3 x 10, 5 x 10) in addition to RegretNet's sensitivity to hyperparameter search \cite{Rahme20:Auction} suggest that these problems can be resolved at the cost of greater computational resources.

\subsection{Non-uniform Fairness}

\begin{table*}[t]
    \centering
    \begin{tabular}{|c|c|c|c|c|c|}                  \hline
    \multicolumn{6}{|c|}{Setting D Revenue: Mean (StDev)} \\ \hline
    b    & \multicolumn{5}{|c|}{d}               \\ \hline
    -    & 1.00  & 0.75  &  0.50 & 0.25  & 0.00  \\ \hline
    0.00 & 2.043 (0.396) & 2.018 (0.381) & 1.999 (0.376) & 1.992 (0.373) & 1.989 (0.376) \\ \hline
    0.25 & 2.544 (0.397) & 2.515 (0.383) & 2.486 (0.376) & 2.489 (0.375) & 2.494 (0.377) \\ \hline
    0.50 & 3.037 (0.395) & 3.016 (0.384) & 2.990 (0.376) & 2.988 (0.376) & 2.997 (0.379) \\ \hline
    0.75 & 3.540 (0.394) & 3.515 (0.384) & 3.489 (0.376) & 3.487 (0.376) & 3.493 (0.380) \\ \hline
    1.00 & 4.037 (0.392) & 4.010 (0.383) & 3.987 (0.375) & 3.988 (0.377) & 3.996 (0.379) \\ \hline
    \end{tabular}
     \caption{Revenue for Setting D -- a 3 bidder x 4 item auction with items 1 and 2 valued at $U[0, 1]$ and items 3 and 4 at $U[0, 1] + b$. Fair allocations are enforced for items pairs (1,3) and (2,4).}
    \label{tab:drevenue}
\end{table*}

\begin{table*}[t]
    \centering
    \begin{tabular}{|c|c|c|c|c|c|}                  \hline
    \multicolumn{6}{|c|}{Setting E Revenue: Mean (StDev)} \\ \hline
    b    & \multicolumn{5}{|c|}{d}               \\ \hline
    -    & 1.00  & 0.75  &  0.50 & 0.25  & 0.00  \\ \hline
    0.00 & 2.043 (0.396) & 2.031 (0.389) & 2.039 (0.392) & 2.041 (0.392) & 2.033 (0.393) \\ \hline
    0.25 & 2.544 (0.397) & 2.521 (0.388) & 2.528 (0.390) & 2.537 (0.394) & 2.525 (0.393) \\ \hline
    0.50 & 3.037 (0.395) & 3.025 (0.385) & 3.005 (0.384) & 3.006 (0.383) & 3.010 (0.385) \\ \hline
    0.75 & 3.540 (0.394) & 3.517 (0.386) & 3.505 (0.383) & 3.511 (0.384) & 3.509 (0.385) \\ \hline
    1.00 & 4.037 (0.392) & 4.019 (0.381) & 3.999 (0.382) & 4.015 (0.382) & 4.006 (0.385) \\ \hline
    \end{tabular}
     \caption{Revenue for Setting E -- identical to setting D but fair allocations are enforced for items 2, 3, and 4. Interestingly, there is a slight increase in revenue from Setting D.}
    \label{tab:erevenue}
\end{table*}

\begin{table*}[t]
    \centering
    \begin{tabular}{|c|c|c|c|c|c|}                  \hline
    \multicolumn{6}{|c|}{Setting F Revenue: Mean (StDev)} \\ \hline
    b    & \multicolumn{5}{|c|}{d}               \\ \hline
    -    & 1.00  & 0.75  &  0.50 & 0.25  & 0.00  \\ \hline
    0.00 & 2.043 (0.396) & 2.022 (0.385) & 2.004 (0.377) & 1.992 (0.373) & 1.992 (0.375) \\ \hline
    0.25 & 2.544 (0.397) & 2.525 (0.387) & 2.494 (0.377) & 2.492 (0.375) & 2.495 (0.375) \\ \hline
    0.50 & 3.037 (0.395) & 3.022 (0.386) & 2.991 (0.375) & 2.987 (0.376) & 2.994 (0.377) \\ \hline
    0.75 & 3.540 (0.394) & 3.514 (0.382) & 3.494 (0.376) & 3.492 (0.375) & 3.493 (0.377) \\ \hline
    1.00 & 4.037 (0.392) & 4.012 (0.384) & 3.988 (0.379) & 3.992 (0.376) & 3.993 (0.377) \\ \hline
    \end{tabular}
     \caption{Revenue for Setting F -- identical to setting D but one of the bidders is given fair allocations for user pairs (1,2) and (3,4).}
    \label{tab:frevenue}
\end{table*}

In addition to the uniform fairness experiments above, we define three simple settings to investigate enforcement of fairness in cases where preferences may be more complex. Here, we consider auctions with three bidders and four items denoted $u_{1...4}$. Each item has binary features $f_1$ and $f_2$ which are used to construct the distance values used in our fairness constraint -- in an ad auction setting, these could be characteristics of different groups of users.

D. The parity of users with respect to both $f_1$ and $f_2$ is balanced. All bidders are constrained by a distance metric on $f_2$.

E. The parity of users with respect to $f_2$ is imbalanced. All bidders are constrained by a distance metric on $f_2$.

F. Identical to setting C, but bidder 3 is constrained by a distance metric on $f_1$ rather than $f_2$.

\begin{table}
    \centering
    \begin{tabular}{c|c|c}
         D, F  & $f_1$ & $f_2$  \\ \hline
         $u_1$ & 0 & 0 \\
         $u_2$ & 0 & 1 \\
         $u_3$ & 1 & 0 \\
         $u_4$ & 1 & 1 
    \end{tabular}
    \quad
    \begin{tabular}{c|c|c}
         E     & $f_1$ & $f_2$ \\ \hline
         $u_1$ & 0 & 1 \\
         $u_2$ & 0 & 1 \\
         $u_3$ & 1 & 0 \\
         $u_4$ & 1 & 1 
    \end{tabular}
    \caption{User features for settings D, E, and F.}
    \label{tab:cdfeatures}
\end{table}

Suppose that $f_2$ is a ``relevant'' trait (e.g. advertisers are showing software engineer hiring ads, and $f_2$ denotes whether or not a user has a computer science degree). Since the user pairs $u_{1,3}$ and $u_{2,4}$ have matching values of $f_2$, we are interested in treating them similarly using the distance function
$$D_{j,j'}(d) = 1 - (1-d)(1 - |f_2(u_j) - f_2(u_{j'})|) , $$
where the parameter $d \in [0, 1]$ adjusts the level of fairness. A value of $d=1$ has no consideration for fairness, while $d=0$ requires users with matching values of $f_2$ to have identical allocations. Note that in setting E, the fairness constraint applied to bidder 3 uses $f_1$ rather than $f_2$.

We are interested in how ProportionNet handles bidders who bid higher on users with the $f_1$ feature. If $f_1$ denotes a feature (e.g. gender) which is irrelevant to the auction (software engineer hiring ads), this bidding behavior can be viewed as discriminatory. The new valuation range is $[0, 1] + b|f_1(u_j) - f_1(u_j')|$ for all three bidders, where $b \in \mathbb{R}$ is a parameter that adjusts the level of discriminatory bidding behavior. We tested a grid of values $b, d = 0, 0.25, ..., 1.00$. Tables \ref{tab:drevenue}, \ref{tab:erevenue}, and \ref{tab:frevenue} show the revenue results; information on unfairness and regret is in the appendix.

\section{Ethical Impact}

Work in the field of automated mechanism design, including recent work like the RegretNet approach of \cite{dutting2019optimal}, show that tools and techniques from machine learning can help address persistent challenges in the theory of auctions.
Our work builds on this body of research by positing the tools of machine learning can address another problem in auctions: fairness considerations on the item side. 
Above, we have shown compelling empirical evidence that support this claim based on the addition of fairness constraints to RegretNet's augmented Lagrangian technique.

One of the major social problems associated with online advertising is its use in the job market. The value proposition of online advertising often involves targeting ads to specific demographic groups, and this is a serious problem when those groups may represent, even indirectly, protected classes.
At least in the United States, antidiscrimination law is codified in Title VII of the Civil Rights Act of 1964 which limits the type of behaviors employers can engage in. 
The US Supreme Court decided in Griggs v. Duke Power Co. \cite{Griggs1971}, that certain behaviors which might cause discriminatory results, even if they are performed unintentionally, are illegal -- the doctrine of ``disparate impact''.
The Court has shied away from rigidly defining disparate impact (in quantitative terms) \cite{Ricci2009}, but the Equal Employment Opportunity Commission (EEOC) makes determinations about disparate impact based on the 80\% rule \cite{EEOC1978}. 
This rule generally states that a group of individuals in a protected class cannot have a selection rate less than 80\% of the highest rate for another class.
Since this generally can be mathematically formalized, it has been studied in the fair ML literature from a practical and technical perspective \cite{barocas2016big, feldman2015certifying}. Our definition of fairness does not map directly onto the 80\% rule, but it shares some similarities; when the distance metric is defined in terms of protected classes, it arguably constrains allocations from having ``disparate impact''.

While in some sense, we thus provide a way to learn a mechanism that will satisfy widely-held definitions of fairness, to view our proposed approach as a cure-all would be misguided. There is a growing body of work that shows there is significant daylight between how a computer scientist thinks about fairness and how others do.
\citet{holstein2019improving} conducted interviews with developers about their desired fairness outcomes and showed them to sometimes be at odds with each other. 
\citet{saha2019measuring} demonstrated that laypeople often don't comprehend computer science notions of fairness as well.
Then there are also more sociological critiques of the general fair ML approaches writ large \cite{selbst2019fairness, hanna2020towards}.
As \citet{selbst2019fairness} and \cite{hutchinson201950} point out, the reification of fairness concepts into mathematical formulae has inherent problems.
Fair ML has achieved prominence through the translation of nebulous and debatable definitions into concrete mathematics, often providing a veneer of objectivity over highly contested notions of equality and justice.

We acknowledge that our work at its present stage is a technical intervention, rather than an analysis or critique of a sociotechnical system. 
We do not aim to make prescriptive statements on our own about the ultimately correct way to design auctions that must be fair, as that is best done in an interdisciplinary group with multiple stakeholders. 
However, we hope that our work can be one useful contribution to the challenging problem of fair mechanism design.

\section{Conclusions \& Future Technical Work}\label{sec:conclusions}

Future work might include incorporating improvements to the training procedure as in \cite{Rahme20:Auction}, or making use of techniques that can exactly evaluate the degree to which strategyproofness is violated, as in \cite{Curry20:Certifying}.

Additionally, the theoretical question of characterizing which fair, strategyproof mechanisms maximize revenue is an interesting one. \cite{Celis19:Toward} has provided some useful work in this direction already, for a specific class of auctions and notion of fairness. Perhaps the use of deep-learning-based techniques to approximate fully general multi-item mechanisms can provide a starting point for theory as has happened in \cite{dutting2019optimal}.

Finally, while our work is motivated by the problem of unfairness in advertising auctions, our models are still quite stylized. Enhancing the realism of the model with real-world data, valuations, or fairness constraints derived from real settings would be quite interesting.

\input{acknowledgments}

\bibliographystyle{ACM-Reference-Format}
\bibliography{bibliography}

\clearpage

\input{appendix}

\end{document}

%% file: comment.tex
\usepackage{xspace}

\usepackage{xcolor}

\newcommand{\ignore}[1]{}

%% file: acknowledgments.tex
\section*{Acknowledgments}
Curry, Dickerson, and Dooley were supported in part by NSF CAREER Award IIS-1846237, NIST MSE Award \#20126334, DARPA GARD \#HR00112020007, DARPA SI3-CMD \#S4761, DoD WHS Award \#HQ003420F0035, and a Google Faculty Research Award.
Chiang and Goldstein were supported by the AFOSR MURI Program, DARPA GARD and DARPA QED4RML programs.
Horishny, Kuo, and Ostuni were supported by NSF Award CCF-1852352 through the University of Maryland's REU-CAAR: Combinatorics
and Algorithms Applied to Real Problems.  We thank Bill Gasarch for his standing commitment to building and maintaining a strong REU program at the University of Maryland.

%% file: appendix.tex
\section{Appendix}

\subsection{Generalization Bound for Unfairness}
We restate the theorem below:

\bound*

\begin{proof}
Let $\mathcal{G}_j$ be the class of \textbf{item-wise} allocation functions for item $j$ defined on a class of auctions $\mathcal{M}$.

Let $\unf_j \circ G$ be the class of \textbf{unfairness functions} for item $j$. A function $f_j \in \unf_j \circ G$ maps $f_j: V \rightarrow \mathbb{R}$. Extended to all items, $\unf \circ G$ is the class of tuples $(f_1, ..., f_m)$. Such a vector-valued function $f \in \unf \circ G$ maps $f: V \rightarrow \mathbb{R}^m$. Finally, we also define the class of \textbf{sum unfairness functions}: 
\begin{dmath*}
\overline{\unf} \circ G = \{{f: V \rightarrow \mathbb{R} }\mid {f(v) = \sum_{j=1}^{m} f_j(v)} \text{ for some }{(f_1, ... f_m) \in \unf \circ G_j}\}
\end{dmath*}

We prove bounds for the simple case of a uniform distance constraint $d$ between all users, with all bidders in one advertising category. In this case, given a mechanism $(g, p)$, the quantity for item $j$'s unfairness is:
\begin{dmath*}\unf_j(v) = \sum_{{j'}\in M}\max(0, \sum_{i\in N}\max(0, g_{i,j}(v) - g_{i,j'}(v)) - d)\end{dmath*}

Our proof hews very closely to the generalization bound for regret in D.2.4 of \cite{dutting2019optimal}. We use the same notion of $\ell_{\infty, 1}$ distance between functions and, under this distance, relate covering numbers of the function classes defined above:

\begin{dmath*} \mathcal{N}_{\infty}(\overline{\unf} \circ \mathcal{G}, \epsilon) 
\leq 
\mathcal{N}_{\infty}(\unf \circ \mathcal{G}, \frac{\epsilon}{m})
\leq
\mathcal{N}_{\infty}(\mathcal{G},\frac{\epsilon}{2m^3})
\leq
\mathcal{N}_{\infty}(\mathcal{M},\frac{\epsilon}{2m^3})
\end{dmath*}

These covering numbers in turn bound empirical Rademacher complexity, which allows us to apply the same lemma of \cite{shalev2014understanding} used in \cite{dutting2019optimal}.

\subsubsection{Step 1. Bounding $\mathcal{N}_{\infty}(\mathcal{G}, \epsilon) \leq \mathcal{N}_{\infty}(\mathcal{M}, \epsilon)$}

By the definition of $\mathcal{N}_{\infty}(\mathcal{M}, \epsilon)$, there exists some cover $\hat{\mathcal{M}}$ (where $\left| \mathcal{\hat{M}} \right|\leq\mathcal{N}_{\infty}(\mathcal{M}, \epsilon)$) such that $\forall (g,p) \in \mathcal{M}, \exists (\hat{g}, \hat{p}) \in \mathcal{\hat{M}}$ where 
\begin{dmath*}
\sup_{v\in V} \sum_{i,j} 
\left| g_{i,j}(v) - \hat{g}_{i,j}(v) \right| + 
\lVert p(v) - \hat{p}(v) \rVert_1 
\leq \epsilon.
\end{dmath*}

It is trivial to bound the distance between any $g \in \mathcal{G}$ and its covering $\hat{g} \in \mathcal{\hat{G}},\ \forall v\in V$:

\begin{dmath*}\sum_{j\in M} \left\| g_{\cdot,j} - \hat{g}_{\cdot,j} \right\|_1 = \sum_{i,j} \left| g_{i,j}(v) - \hat{g}_{i,j}(v)\right| \leq \sum_{i,j} 
\left| g_{i,j}(v) - \hat{g}_{i,j}(v) \right| + 
\lVert p(v) - \hat{p}(v) \rVert_1 \leq \epsilon.\end{dmath*}
Therefore, $\mathcal{N}_{\infty}(\mathcal{G}, \epsilon) \leq \mathcal{N}_{\infty}(\mathcal{M}, \epsilon)$

\subsubsection{Step 2. Bounding $\mathcal{N}_{\infty}(\unf \circ \mathcal{G}, \epsilon) \leq \mathcal{N}_{\infty}(\mathcal{G}, \frac{\epsilon}{2m^2})$}

We first bound $\mathcal{N}_{\infty}(\unf_j \circ \mathcal{G}, \epsilon) \leq \mathcal{N}_{\infty}(\mathcal{G}, \frac{\epsilon}{2m})$ for a single $j$.

Taking $g, \hat{g}$ satisfying the definition of $\mathcal{N}_\infty(\mathcal{G}, \frac{\epsilon}{2m})$ and fixing a fairness parameter $d \in [0, 1]$, we bound the $\ell_{\infty, 1}$ distance between $\unf_j \circ g$ and $\unf_j \circ \hat{g}$. Note that $g_j(v)$ has been shortened to $g_j$, $g_{i,j}(v)$ to $g_{i,j}$, etc. for convenience. We use the fact that $\left|\max(0, a) - \max(0, b)\right| \leq \left| a - b \right|$.
\begin{dgroup*}
    \begin{dmath*}
    \sup_{v\in V} \left| 
\unf_j \circ g - \unf_j \circ \hat{g}
\right|\end{dmath*}
 \begin{dmath*}= \sup_{v\in V}
\left| 
\sum_{j' \in M} \max(0, \left(\sum_{i\in N}\max(0, g_{i,j} - g_{i,j'})\right) - d) - 
\sum_{j' \in M} \max(0, \left(\sum_{i\in N}\max(0, \hat{g}_{i,j} - \hat{g}_{i,j'})\right) - d)
\right| \end{dmath*}
\begin{dmath*}
\leq \sup_{v\in V} \sum_{j' \in M}
\left| 
\max(0, \left(\sum_{i\in N}\max(0, g_{i,j} - g_{i,j'})\right) - d) - 
\max(0, \left(\sum_{i\in N}\max(0, \hat{g}_{i,j} - \hat{g}_{i,j'})\right) - d)
\right| \end{dmath*}
\begin{dmath*}\leq \sup_{v\in V} \sum_{j' \in M}
\left| \sum_{i\in N}\max(0, g_{i,j} - g_{i,j'}) - 
\sum_{i\in N}\max(0, \hat{g}_{i,j} - \hat{g}_{i,j'})
\right| \end{dmath*}
\begin{dmath*}\leq \sup_{v\in V} \sum_{j' \in M} \sum_{i \in N}
\left| \max(0, g_{i,j} - g_{i,j'}) - 
\max(0, \hat{g}_{i,j} - \hat{g}_{i,j'})
\right| \end{dmath*}
\begin{dmath*}\leq \sup_{v\in V} \sum_{j' \in M} \sum_{i \in N}
\left| g_{i,j} - g_{i,j'} - 
\hat{g}_{i,j} + \hat{g}_{i,j'} \right| \end{dmath*}
\begin{dmath*}= \sup_{v\in V} \sum_{j' \in M} \sum_{i \in N}
\left| (g_{i,j} -
\hat{g}_{i,j}) - (g_{i,j'} -\hat{g}_{i,j'})  \right| 
\end{dmath*}
\begin{dmath*}\leq \sup_{v\in V} \sum_{j' \in M}\sum_{i \in N}
\left| g_{i,j} -
\hat{g}_{i,j}\right| + \sup_{v\in V}\sum_{j' \in M} \sum_{i \in N}\left| g_{i,j'} -\hat{g}_{i,j'}  \right|
\end{dmath*}
\begin{dmath*}= \sup_{v\in V} m \cdot \sum_{i \in N}
\left| g_{i,j} -
\hat{g}_{i,j}\right| + \sup_{v\in V}\sum_{j' \in M} \sum_{i \in N}\left| g_{i,j'} -\hat{g}_{i,j'}  \right|
\end{dmath*}
\begin{dmath*}\leq \sup_{v\in V} m \cdot \sum_{j \in M}\sum_{i \in N}
\left| g_{i,j} -
\hat{g}_{i,j}\right| + \sup_{v\in V}\sum_{j' \in M} \sum_{i \in N}\left| g_{i,j'} -\hat{g}_{i,j'}  \right|
\leq \frac{(m+1)\epsilon}{2m} \leq \epsilon\end{dmath*}
\end{dgroup*}

where the second-to-last inequality follows due to the definition of the cover to which $\hat{g}_{i,j}, \hat{g}_{i,j'}$ belong.

Therefore, $\mathcal{N}_{\infty}(\unf_j \circ \mathcal{G}, \epsilon) \leq \mathcal{N}_{\infty}(\mathcal{G}, \frac{\epsilon}{2m})$, which implies $\mathcal{N}_{\infty}(\unf \circ \mathcal{G}, \epsilon) \leq \mathcal{N}_{\infty}(\mathcal{G}, \frac{\epsilon}{2m^2})$

\subsubsection{Step 3. Bounding $\mathcal{N}(\overline{\unf} \circ \mathcal{G},\epsilon) \leq \mathcal{N}(\unf \circ \mathcal{G},\frac{\epsilon}{m})$}

We take $g_j, \hat{g}_j$ from a cover of $\unf \circ \mathcal{G}$. The $\ell_\infty$ distance between $\overline{\unf} \circ g$ and $\overline{\unf} \circ \hat{g}$ is:
\begin{dmath*}
\sup_{v\in V} \left|
\overline{\unf} \circ g - \overline{\unf} \circ \hat{g}
\right| 
\leq \sum_{j \in M} \sup_{v \in V} \left| \unf_j \circ g - \unf_j \circ \hat{g} \right| \leq 
\sum_{j\in M}\frac{\epsilon}{m}
= \epsilon,
\end{dmath*}

Therefore, $\mathcal{N}_{\infty}(\overline{\unf} \circ \mathcal{G}, \epsilon) \leq \mathcal{N}_{\infty}(\unf \circ \mathcal{G}, \frac{\epsilon}{m}).$

Combining these inequalities, we get $\mathcal{N}_{\infty}(\overline{\unf} \circ \mathcal{G}, \epsilon) \leq \mathcal{N}_{\infty}(\mathcal{M}, \frac{\epsilon}{2m^3})$ as desired.

\subsubsection{Applying bounds}

For convenience denote $\mathcal{F} = \overline{\unf} \circ \mathcal{G}$, with $\hat{\mathcal{F}}$ as its cover. Denote by $\hat{f}_f \in \hat{\mathcal{F}}$ the closest covering point to some $f \in \mathcal{F}$.

We wish to bound the empirical Rademacher complexity $\hat{\mathcal{R}}_L(\mathcal{F})$, which is 
\begingroup
\allowdisplaybreaks
\begin{align*}
     & \frac{1}{L} \mathbb{E}_\sigma \left[ \sup_{f \in \mathcal{F}} \sum_{\ell = 1}^L \sigma_\ell f(v^\ell) \right] \\
     &= \frac{1}{L} \mathbb{E}_\sigma \left[ \sup_{f \in \mathcal{F}} \sum_{\ell = 1}^L \sigma_\ell \left(\hat{f}_f(v^\ell) +  f(v^\ell) - \hat{f}_f(v^\ell) \right)\right] \\
     &\leq \frac{1}{L} \mathbb{E}_\sigma \left[ \sup_{f \in \mathcal{F}} \sum_{\ell = 1}^L \sigma_\ell \hat{f}_f(v^\ell)\right] + \frac{1}{L}\mathbb{E}_\sigma \left[ \sup_{f \in \mathcal{F}} \sum_{\ell = 1}^L  \left(f(v^\ell) - \hat{f}_f(v^\ell)\right) \right]  \\
      &\leq \frac{1}{L} \mathbb{E}_\sigma \left[ \sup_{f \in \mathcal{F}} \sum_{\ell = 1}^L \sigma_\ell \hat{f}_f(v^\ell)\right] + \frac{1}{L}\mathbb{E}_\sigma \left[  \sum_{\ell = 1}^L  \sup_{f \in \mathcal{F}} \left(f(v^\ell) - \hat{f}_f(v^\ell)\right) \right] \\
      &\leq \frac{1}{L} \mathbb{E}_\sigma \left[ \sup_{f \in \mathcal{F}} \sum_{\ell = 1}^L \sigma_\ell \hat{f}_f(v^\ell)\right] + \epsilon \\
        &= \frac{1}{L} \mathbb{E}_\sigma \left[ \sup_{\hat{f} \in \mathcal{\hat{f}}} \sum_{\ell = 1}^L \sigma_\ell \hat{f}(v^\ell)\right] + \epsilon \\
    &\leq \frac{1}{L} \max_{\hat{f} \in \hat{\mathcal{F}}} \sqrt{\sum_{\ell=1}^L \hat{f}(v^\ell)^2} \sqrt{2 \log \mathcal{N}_{\infty}(\mathcal{F},\epsilon)} + \epsilon \text{ (by Massart's lemma)}\\
    &\leq \frac{1}{L} \max_{f \in \mathcal{F}} \sqrt{\sum_{\ell=1}^L \left(f(v^\ell) + \epsilon\right)^2} \sqrt{2 \log \mathcal{N}_{\infty}(\mathcal{F},\epsilon)} + \epsilon \\
    &\leq \frac{1}{L} \sqrt{\sum_{\ell=1}^L \left(nm^2 + \epsilon\right)^2} \sqrt{2 \log \mathcal{N}_{\infty}(\mathcal{F},\epsilon)} + \epsilon \\
    &\leq \frac{1}{L} \sqrt{\sum_{\ell=1}^L \left(nm^2 + \epsilon\right)^2} \sqrt{2 \log \mathcal{N}_{\infty}(\mathcal{F},\epsilon)} + \epsilon \\
    &=\left(nm^2 + \epsilon\right) \sqrt{\frac{2 \log \mathcal{N}_{\infty}(\mathcal{F},\epsilon)}{L}} + \epsilon
\end{align*}

Given $\mathcal{N}_\infty(\mathcal{F}, \epsilon) \leq \mathcal{N}_\infty(\mathcal{M}, \frac{\epsilon}{2m})$, we can then apply the lemma of \cite{shalev2014understanding} as in \cite{dutting2019optimal} to say that with probability $1 - \delta$, for a distribution-independent constant $C$:
\begin{align*}
    \mathbb{E}_{v}\left[f(v)\right] &\leq \frac{1}{L} \sum_{\ell=1}^L f(v^\ell) + 2 \hat{\mathcal{R}}_L(\mathcal{F}) + 4C\sqrt{\frac{2\log(4/\delta)}{L}} \\
                &\leq \frac{1}{L} \sum_{\ell=1}^L f(v^\ell) + \\
                &2 \inf_{\epsilon > 0}\left((nm^2 + \epsilon)\sqrt{\frac{2 \log \mathcal{N}_\infty(\mathcal{M}, \frac{\epsilon}{2m^3})}{L}} + \epsilon\right)\\
                &+ 4C\sqrt{\frac{2\log(4/\delta)}{L}}
\end{align*}
\endgroup
\end{proof}

\subsection{Additional Results}

The following tables show the mean and standard deviation of regret and unfairness for all configurations of Setting C (Table \ref{tab:cregretunfair}), D (Table \ref{tab:dregretunfair}), E (Table \ref{tab:eregretunfair}), and F (Table \ref{tab:fregretunfair}). We achieve low values for all of these quantities, which is desirable. A non-vanishing value as the network is trained implies that the auction is empirically non-strategyproof (if regret cannot be minimized) or unfair (if unfairness cannot be minimized).

\begin{table*}
    \centering
\begin{tabular}{|c|c|c|c|c|c|c|}
\hline
\multicolumn{7}{|c|}{Sweep Regret, Mean (StDev)} \\ \hline
n x m & $\ell$ & \multicolumn{5}{|c|}{D}          \\ \hline
-     & - & 1.00 & 0.75 & 0.50 & 0.25 & 0.00 \\
\hline
1 x 2 & 2 & 0.000 (0.000) & 0.000 (0.001) & 0.000 (0.000) & 0.000 (0.000) & 0.001 (0.001) \\
1 x 3 & 2 & 0.001 (0.001) & 0.001 (0.001) & 0.000 (0.000) & 0.001 (0.001) & 0.001 (0.000) \\
1 x 4 & 2 & 0.001 (0.001) & 0.000 (0.001) & 0.001 (0.001) & 0.001 (0.001) & 0.001 (0.001) \\
1 x 5 & 2 & 0.001 (0.002) & 0.001 (0.001) & 0.001 (0.001) & 0.001 (0.001) & 0.001 (0.000) \\
1 x 6 & 2 & 0.001 (0.001) & 0.001 (0.001) & 0.001 (0.001) & 0.001 (0.001) & 0.001 (0.000) \\
\hline
2 x 2 & 2 & 0.001 (0.001) & 0.001 (0.001) & 0.001 (0.001) & 0.001 (0.001) & 0.002 (0.002) \\
2 x 3 & 2 & 0.002 (0.003) & 0.001 (0.002) & 0.002 (0.003) & 0.001 (0.001) & 0.001 (0.001) \\
2 x 4 & 2 & 0.002 (0.003) & 0.002 (0.003) & 0.001 (0.002) & 0.001 (0.002) & 0.002 (0.002) \\
2 x 5 & 3 & 0.003 (0.004) & 0.002 (0.003) & 0.004 (0.006) & 0.001 (0.002) & 0.001 (0.001) \\
2 x 6 & 3 & 0.003 (0.004) & 0.003 (0.004) & 0.002 (0.003) & 0.003 (0.005) & 0.002 (0.003) \\
\hline
3 x 2 & 2 & 0.002 (0.002) & 0.002 (0.002) & 0.001 (0.001) & 0.001 (0.002) & 0.001 (0.001) \\
3 x 3 & 2 & 0.003 (0.003) & 0.003 (0.003) & 0.001 (0.002) & 0.002 (0.002) & 0.001 (0.001) \\
3 x 4 & 3 & 0.006 (0.006) & 0.004 (0.005) & 0.002 (0.003) & 0.002 (0.002) & 0.002 (0.002) \\
3 x 5 & 3 & 0.006 (0.007) & 0.005 (0.006) & 0.002 (0.005) & 0.002 (0.003) & 0.003 (0.003) \\
3 x 6 & 3 & 0.009 (0.011) & 0.006 (0.007) & 0.004 (0.005) & 0.004 (0.005) & 0.002 (0.002) \\
\hline
4 x 2 & 2 & 0.003 (0.002) & 0.003 (0.002) & 0.003 (0.002) & 0.002 (0.002) & 0.001 (0.001) \\
4 x 3 & 3 & 0.004 (0.003) & 0.004 (0.002) & 0.003 (0.003) & 0.002 (0.002) & 0.001 (0.001) \\
4 x 4 & 3 & 0.006 (0.004) & 0.006 (0.004) & 0.002 (0.002) & 0.003 (0.003) & 0.002 (0.003) \\
4 x 5 & 4 & 0.007 (0.005) & 0.007 (0.004) & 0.004 (0.002) & 0.003 (0.004) & 0.003 (0.003) \\
4 x 6 & 4 & 0.010 (0.076) & 0.007 (0.006) & 0.006 (0.006) & 0.006 (0.007) & 0.003 (0.003) \\
\hline
5 x 2 & 3 & 0.003 (0.003) & 0.005 (0.004) & 0.005 (0.003) & 0.004 (0.003) & 0.005 (0.006) \\
5 x 3 & 3 & 0.006 (0.003) & 0.010 (0.005) & 0.008 (0.004) & 0.006 (0.004) & 0.003 (0.003) \\
5 x 4 & 4 & 0.007 (0.004) & 0.011 (0.005) & 0.006 (0.008) & 0.003 (0.004) & 0.003 (0.004) \\
5 x 5 & 4 & 0.003 (0.045) & 0.008 (0.007) & 0.007 (0.008) & 0.005 (0.010) & 0.002 (0.003) \\
5 x 6 & 5 & 0.190 (0.377) & 0.005 (0.008) & 0.005 (0.015) & 0.005 (0.010) & 0.003 (0.004) \\
\hline

\multicolumn{7}{|c|}{Sweep Unfairness, Mean (StDev)} \\ \hline
n x m & $\ell$ & \multicolumn{5}{|c|}{D}          \\ \hline
-     & - & 1.00 & 0.75 & 0.50 & 0.25 & 0.00 \\
\hline
1 x 2 & 2 & 0.000 (0.000) & 0.000 (0.000) & 0.000 (0.000) & 0.000 (0.000) & 0.002 (0.002) \\
1 x 3 & 2 & 0.000 (0.000) & 0.000 (0.001) & 0.000 (0.004) & 0.000 (0.002) & 0.005 (0.010) \\
1 x 4 & 2 & 0.000 (0.000) & 0.000 (0.000) & 0.000 (0.004) & 0.000 (0.005) & 0.009 (0.018) \\
1 x 5 & 2 & 0.000 (0.000) & 0.000 (0.009) & 0.000 (0.008) & 0.000 (0.007) & 0.006 (0.015) \\
1 x 6 & 2 & 0.000 (0.000) & 0.000 (0.006) & 0.000 (0.003) & 0.000 (0.009) & 0.009 (0.021) \\
\hline
2 x 2 & 2 & 0.000 (0.000) & 0.000 (0.003) & 0.000 (0.003) & 0.000 (0.001) & 0.006 (0.010) \\
2 x 3 & 2 & 0.000 (0.000) & 0.000 (0.003) & 0.000 (0.003) & 0.000 (0.000) & 0.005 (0.009) \\
2 x 4 & 2 & 0.000 (0.000) & 0.000 (0.006) & 0.000 (0.001) & 0.000 (0.000) & 0.015 (0.023) \\
2 x 5 & 3 & 0.000 (0.000) & 0.000 (0.007) & 0.000 (0.000) & 0.000 (0.000) & 0.008 (0.011) \\
2 x 6 & 3 & 0.000 (0.000) & 0.003 (0.043) & 0.001 (0.017) & 0.000 (0.000) & 0.031 (0.047) \\
\hline
3 x 2 & 2 & 0.000 (0.000) & 0.001 (0.006) & 0.001 (0.007) & 0.000 (0.006) & 0.004 (0.008) \\
3 x 3 & 2 & 0.000 (0.000) & 0.000 (0.007) & 0.000 (0.001) & 0.000 (0.000) & 0.008 (0.013) \\
3 x 4 & 3 & 0.000 (0.000) & 0.000 (0.008) & 0.000 (0.002) & 0.000 (0.000) & 0.012 (0.016) \\
3 x 5 & 3 & 0.000 (0.000) & 0.001 (0.019) & 0.002 (0.031) & 0.000 (0.000) & 0.023 (0.034) \\
3 x 6 & 3 & 0.000 (0.000) & 0.001 (0.011) & 0.001 (0.015) & 0.001 (0.023) & 0.017 (0.024) \\
\hline
4 x 2 & 2 & 0.000 (0.000) & 0.001 (0.003) & 0.000 (0.004) & 0.000 (0.001) & 0.004 (0.006) \\
4 x 3 & 3 & 0.000 (0.000) & 0.001 (0.014) & 0.001 (0.013) & 0.000 (0.001) & 0.009 (0.012) \\
4 x 4 & 3 & 0.000 (0.000) & 0.001 (0.011) & 0.000 (0.000) & 0.000 (0.000) & 0.016 (0.020) \\
4 x 5 & 4 & 0.000 (0.000) & 0.000 (0.006) & 0.000 (0.000) & 0.000 (0.000) & 0.022 (0.033) \\
4 x 6 & 4 & 0.000 (0.000) & 0.001 (0.014) & 0.000 (0.010) & 0.002 (0.022) & 0.020 (0.026) \\
\hline
5 x 2 & 3 & 0.000 (0.000) & 0.000 (0.001) & 0.000 (0.003) & 0.000 (0.002) & 0.008 (0.013) \\
5 x 3 & 3 & 0.000 (0.000) & 0.008 (0.022) & 0.010 (0.032) & 0.003 (0.024) & 0.018 (0.026) \\
5 x 4 & 4 & 0.000 (0.000) & 0.001 (0.006) & 0.000 (0.000) & 0.000 (0.000) & 0.008 (0.013) \\
5 x 5 & 4 & 0.000 (0.000) & 0.000 (0.003) & 0.000 (0.002) & 0.000 (0.003) & 0.007 (0.010) \\
5 x 6 & 5 & 0.000 (0.000) & 0.000 (0.000) & 0.000 (0.000) & 0.000 (0.000) & 0.011 (0.016) \\
\hline
\end{tabular}
\caption{Regret and unfairness for statistics for Setting C. Note that for the 5 x 6 setting, the network converged to a suboptimal solution with both high payment and high regret.}
\label{tab:cregretunfair}
\end{table*}

\begin{table*}[t]
    \centering
    \begin{tabular}{|c|c|c|c|c|c|}                  \hline
    \multicolumn{6}{|c|}{Setting D Regret: Mean (StDev)} \\ \hline
    b    & \multicolumn{5}{|c|}{d}               \\ \hline
    -    & 1.00  & 0.75  &  0.50 & 0.25  & 0.00  \\ \hline
    0.00 & 0.004 (0.006) & 0.004 (0.004) & 0.002 (0.003) & 0.002 (0.002) & 0.001 (0.002) \\ \hline
    0.25 & 0.004 (0.005) & 0.004 (0.004) & 0.002 (0.003) & 0.003 (0.004) & 0.001 (0.001) \\ \hline
    0.50 & 0.005 (0.005) & 0.003 (0.004) & 0.002 (0.003) & 0.002 (0.002) & 0.001 (0.002) \\ \hline
    0.75 & 0.006 (0.008) & 0.003 (0.004) & 0.003 (0.004) & 0.002 (0.002) & 0.002 (0.002) \\ \hline
    1.00 & 0.006 (0.006) & 0.003 (0.004) & 0.004 (0.007) & 0.001 (0.002) & 0.001 (0.002) \\ \hline
    
    \multicolumn{6}{|c|}{Setting D Unfairness: Mean (StDev)} \\ \hline
    -    & 1.00  & 0.75  &  0.50 & 0.25  & 0.00  \\ \hline
    0.00 & 0.000 (0.000) & 0.000 (0.006) & 0.000 (0.003) & 0.000 (0.000) & 0.002 (0.003) \\ \hline
    0.25 & 0.000 (0.000) & 0.000 (0.004) & 0.000 (0.000) & 0.000 (0.000) & 0.001 (0.002) \\ \hline
    0.50 & 0.000 (0.000) & 0.000 (0.005) & 0.000 (0.000) & 0.000 (0.000) & 0.001 (0.002) \\ \hline
    0.75 & 0.000 (0.000) & 0.000 (0.005) & 0.000 (0.000) & 0.000 (0.000) & 0.001 (0.003) \\ \hline
    1.00 & 0.000 (0.000) & 0.035 (0.001) & 0.000 (0.000) & 0.000 (0.000) & 0.001 (0.002) \\ \hline
    \end{tabular}
    \caption{Regret and unfairness statistics for Setting D.}
    \label{tab:dregretunfair}
\end{table*}

\begin{table*}[t]
    \centering
    \begin{tabular}{|c|c|c|c|c|c|}                  \hline
    \multicolumn{6}{|c|}{Setting E Regret: Mean (StDev)} \\ \hline
    b    & \multicolumn{5}{|c|}{d}               \\ \hline
    -    & 1.00  & 0.75  &  0.50 & 0.25  & 0.00  \\ \hline
    0.00 & 0.004 (0.006) & 0.005 (0.006) & 0.003 (0.003) & 0.003 (0.003) & 0.002 (0.003) \\ \hline
    0.25 & 0.004 (0.005) & 0.003 (0.003) & 0.002 (0.003) & 0.003 (0.003) & 0.003 (0.003) \\ \hline
    0.50 & 0.005 (0.005) & 0.005 (0.006) & 0.002 (0.003) & 0.003 (0.004) & 0.003 (0.005) \\ \hline
    0.75 & 0.006 (0.008) & 0.006 (0.007) & 0.002 (0.003) & 0.003 (0.004) & 0.002 (0.004) \\ \hline
    1.00 & 0.006 (0.006) & 0.006 (0.006) & 0.003 (0.004) & 0.003 (0.003) & 0.003 (0.004) \\ \hline
    
    \multicolumn{6}{|c|}{Setting E Unfairness: Mean (StDev)} \\ \hline
    -    & 1.00  & 0.75  &  0.50 & 0.25  & 0.00  \\ \hline
    0.00 & 0.000 (0.000) & 0.000 (0.005) & 0.000 (0.000) & 0.000 (0.000) & 0.002 (0.003) \\ \hline
    0.25 & 0.000 (0.000) & 0.000 (0.006) & 0.000 (0.000) & 0.000 (0.000) & 0.001 (0.003) \\ \hline
    0.50 & 0.000 (0.000) & 0.000 (0.005) & 0.000 (0.000) & 0.000 (0.000) & 0.001 (0.003) \\ \hline
    0.75 & 0.000 (0.000) & 0.000 (0.003) & 0.000 (0.000) & 0.000 (0.000) & 0.002 (0.005) \\ \hline
    1.00 & 0.000 (0.000) & 0.000 (0.000) & 0.000 (0.000) & 0.000 (0.000) & 0.002 (0.005) \\ \hline
    \end{tabular}
    \caption{Regret and unfairness statistics for Setting E. }
    \label{tab:eregretunfair}
\end{table*}

\begin{table*}[t]
    \centering
    \begin{tabular}{|c|c|c|c|c|c|}                  \hline
    \multicolumn{6}{|c|}{Setting F Regret: Mean (StDev)} \\ \hline
    b    & \multicolumn{5}{|c|}{d}               \\ \hline
    -    & 1.00  & 0.75  &  0.50 & 0.25  & 0.00  \\ \hline
    0.00 & 0.004 (0.006) & 0.004 (0.004) & 0.003 (0.004) & 0.003 (0.004) & 0.003 (0.003) \\ \hline
    0.25 & 0.004 (0.005) & 0.005 (0.005) & 0.004 (0.005) & 0.001 (0.002) & 0.002 (0.003) \\ \hline
    0.50 & 0.005 (0.005) & 0.005 (0.005) & 0.003 (0.004) & 0.004 (0.005) & 0.004 (0.004) \\ \hline
    0.75 & 0.006 (0.008) & 0.004 (0.005) & 0.003 (0.004) & 0.002 (0.002) & 0.003 (0.003) \\ \hline
    1.00 & 0.006 (0.006) & 0.004 (0.005) & 0.003 (0.004) & 0.002 (0.002) & 0.003 (0.003) \\ \hline
    
    \multicolumn{6}{|c|}{Setting F Unfairness: Mean (StDev)} \\ \hline
    -    & 1.00  & 0.75  &  0.50 & 0.25  & 0.00  \\ \hline
    0.00 & 0.000 (0.000) & 0.000 (0.005) & 0.000 (0.003) & 0.000 (0.000) & 0.003 (0.006) \\ \hline
    0.25 & 0.000 (0.000) & 0.000 (0.007) & 0.000 (0.002) & 0.000 (0.000) & 0.004 (0.007) \\ \hline
    0.50 & 0.000 (0.000) & 0.000 (0.007) & 0.000 (0.002) & 0.000 (0.000) & 0.007 (0.010) \\ \hline
    0.75 & 0.000 (0.000) & 0.000 (0.004) & 0.000 (0.004) & 0.000 (0.000) & 0.004 (0.006) \\ \hline
    1.00 & 0.000 (0.000) & 0.000 (0.004) & 0.000 (0.004) & 0.000 (0.000) & 0.004 (0.006) \\ \hline
    \end{tabular}
    \caption{Regret and unfairness statistics for Setting F. With the additional fairness category, unfairness becomes slightly harder to minimize.}
    \label{tab:fregretunfair}
\end{table*}